\documentclass[final,nomarks]{dmtcs-episciences}
\usepackage[utf8]{inputenc}
\usepackage{amsmath}

\newtheorem{theorem}{Theorem}[section]
\newtheorem{proposition}[theorem]{Proposition}
\newtheorem{corollary}[theorem]{Corollary}
\newtheorem{lemma}[theorem]{Lemma}

\newtheorem{example}[theorem]{Example}
\newtheorem{remark}[theorem]{Remark}

\newcommand{\Alpha}{A} 
\newcommand{\gmTwo}{2}
\newcommand{\MM}{M}

\DeclareMathOperator{\occur}{occur}
\DeclareMathOperator{\Factor}{Fac}
\DeclareMathOperator{\Prefix}{Pref}
\DeclareMathOperator{\Suffix}{Suf}
\DeclareMathOperator{\CFL}{CFL} 
\DeclareMathOperator{\CFLGL}{CFLGL} 
\DeclareMathOperator{\hiCFL}{hiCFL} 
\DeclareMathOperator{\Start}{S} 
\DeclareMathOperator{\ENW}{ENW} 
\DeclareMathOperator{\BalNaw}{\Omega}
\DeclareMathOperator{\height}{height} 
\DeclareMathOperator{\Balanced}{BAL}
\DeclareMathOperator{\replace}{replace} 
\DeclareMathOperator{\REG}{REG} %

\author[Josef Rukavicka]{Josef Rukavicka}
\title[Dissecting power of intersection of two context-free languages]{Dissecting power of intersection of two context-free languages}
\affiliation{
  Faculty of Nuclear Sciences and Physical Engineering, CTU in Prague, Czech Republic}
\begin{document}
\publicationdata{vol. 25:2 }{2023}{8}{10.46298/dmtcs.9063}{2022-02-08; 2022-02-08; 2022-10-26; 2023-01-26; 2023-06-06}{2023-07-06}
\maketitle

\begin{abstract}
We say that a language $L$ is \emph{constantly growing} if there is a constant $c$ such that for every word $u\in L$ there is a word $v\in L$ with $\vert u\vert<\vert v\vert\leq c+\vert u\vert$.
We say that a language $L$ is \emph{geometrically growing} if there is a constant $c$ such that for every word $u\in L$ there is a word $v\in L$ with $\vert u\vert<\vert v\vert\leq c\vert u\vert$.
Given two infinite languages $L_1,L_2$, we say that $L_1$ \emph{dissects} $L_2$ if $\vert L_2\setminus L_1\vert=\infty$ and $\vert L_1\cap L_2\vert=\infty$. In 2013, it was shown that for every constantly growing language $L$ there is a regular language $R$ such that $R$ dissects $L$.

In the current article we show how to dissect a geometrically growing language by a homomorphic image of intersection of two context-free languages. 

Consider three alphabets $\Gamma$, $\Sigma$, and $\Theta$ such that $\vert \Sigma\vert=1$ and $\vert \Theta\vert=4$. We prove that 
there are context-free languages $\MM_1,\MM_2\subseteq \Theta^*$, an erasing alphabetical homomorphism $\pi:\Theta^*\rightarrow \Sigma^*$, and a nonerasing alphabetical homomorphism $\varphi : \Gamma^*\rightarrow \Sigma^*$ such that: If $L\subseteq \Gamma^*$ is a geometrically growing language then there is a regular language $R\subseteq \Theta^*$ such that $\varphi^{-1}\left(\pi\left(R\cap \MM_1\cap \MM_2\right)\right)$ dissects the language $L$.
\end{abstract}



\section{Introduction}
In the theory of formal languages, the regular and the context-free languages constitute a fundamental concept that attracted a lot of attention in the past several decades. 

In contrast to regular languages, the context-free languages are closed neither under intersection nor under complement. The intersection of context-free languages have been systematically studied; see for instance \cite{GINSBURG1966620,weiner1973,10.1007/978-3-030-40608-0_24}. Let $\CFL_{k}$ denote the family of all languages such that for each $L\in \CFL_k$ there are $k$ context-free languages $L_1, L_2,\dots,L_k$ with $L=\bigcap_{i=1}^kL_i$. For each $k$, it has been shown that there is a language $L\in\CFL_{k+1}$ such that $L\not \in \CFL_{k}$. Thus the $k$-intersections of context-free languages form an infinite hierarchy in the family of all formal languages lying between context-free and context sensitive languages \cite{weiner1973}. 

Dissection of infinite languages belongs to the topics of the theory of formal languages that have been studied in recent years. Let $L_1$ and $L_2$ be infinite languages. We say that $L_1$ \emph{dissects} $L_2$ if $\vert L_2\setminus L_1\vert=\infty$ and $\vert L_1\cap L_2\vert=\infty$. Let $\mathcal{C}$ be a family of languages. We say that a language $L_2$ is $\mathcal{C}$-dissectible  if there is $L_1\in \mathcal{C}$ such that $L_1$ dissects $L_2$. Let $\REG$ denote the family of regular languages. In \cite{YAMAKAMI2013116} the $\REG$-dissectibility has been investigated. Several families of $\REG$-dissectible languages have been presented. Moreover, it has been shown that there are infinite languages that cannot be dissected with a regular language. Also some open questions for $\REG$-dissectibility can be found in \cite{YAMAKAMI2013116}. For example, it is not known if the complement of a context-free language is $\REG$-dissectible.  

Given two countable sets $A$ and $B$, we write $A\subseteq_{ae}B$ if $\vert A-B\vert<\infty$ and we write $A=_{ae}B$ if both $A\subseteq_{ae}B$ and $B\subseteq_{ae}A$ hold. The subscript ``ae'' stands for ``almost everywhere''. We say that \emph{$A$ covers $B$ with an infinite margin} (or \emph{$A$ i-covers $B$}, in short) if both $B\subseteq A$ and $A\not=_{ae}B$ hold. 
We represent a \emph{pair} of languages $A$ and $B$ such that $A$ i-covers $B$ by $i(A,B)$. A language $C$ is said to \emph{separate $i(B,A)$ with infinite margins} (or \emph{i-separates} $i(B,A)$, in short) if $B\subseteq C\subseteq A$, $A\not=_{ae}C$, and $B\not=_{ae}C$. In addition, given two language families $\mathcal{A},\mathcal{B}$ we  define \[i(\mathcal{B},\mathcal{A})=\{i(B,A)\mid A\in\mathcal{A}\mbox{ and }B\in\mathcal{B}\mbox{ and }A\mbox{ i-covers }B\}\mbox{.}\] We say that a language family $\mathcal{C}$ \emph{i-separates} $i(\mathcal{B},\mathcal{A})$ if for every $i(B,A)\in i(\mathcal{B},\mathcal{A})$ there is a language $C\in\mathcal{C}$ such that $C$ i-separates $i(B,A)$. 

Given $\mathcal{A}$ and $\mathcal{B}$ be any two language families, let  \[\mathcal{A}-\mathcal{B}=\{A-B\mid A\in\mathcal{A}\mbox{ and }B\in\mathcal{B}\}\mbox{.}\]
In \cite{YAMAKAMI2013116}, a connection between $\REG$-dissectibility and i-separation has been shown:
\begin{lemma}(see \cite[Lemma $5.1$]{YAMAKAMI2013116})
\label{djueir8887d}
Let $\mathcal{A}$ and $\mathcal{B}$ be any two language families and assume that $\mathcal{A}-\mathcal{B}$ is $\REG$-dissectible. It then holds that, for any $A\in\mathcal{A}$ and any $B\in\mathcal{B}$, if $A$ i-covers $B$, then there exists a language in $\mathcal{E}$ that i-separates $i(B,A)$, where $\mathcal{E}$ expresses the set $\{B\cup(C\cap A)\mid A\in\mathcal{A}, B\in\mathcal{B},C\in\REG\}$. In other words, $\mathcal{E}$ i-separates $i(\mathcal{B}, \mathcal{A})$.
\end{lemma}
Although Lemma \ref{djueir8887d} is stated explicitly for $\REG$-dissectibility, from the proof in \cite{YAMAKAMI2013116} it is clear that any family of languages could be applied. For convenience, in Section \ref{skdjiek99f8f} we present such generalization with a proof; see Lemma \ref{ujdkeijfi4rekd}. This generalized connection between dissectibility and i-separation adds another argument for the study of a dissection of infinite languages. 

There is a longstanding open question in \cite{bucher1980}: Given two context-free languages $L_1,L_2$ such that $L_1\subset L_2$ and $L_2\setminus L_1$ is an infinite language, is there a context-free language $L_3$ such that $L_3\subset L_2$, $L_1\subset L_3$, and both the languages $L_3\setminus L_1$ and $L_2\setminus L_3$ are infinite?  This question was mentioned also in \cite{YAMAKAMI2013116} using the i-separation: Let $\CFL$ denote the family of all context free languages. Does $\CFL$ i-separate $i(\CFL,\CFL)$?
Understanding the dissectibility could help to solve this open question or at least it could help to identify ``minimal'' language families $\mathcal{C}$ such that $\mathcal{C}$ i-separates $i(\CFL,\CFL)$.

Some other results concerning the dissection of infinite languages may be found in \cite{10.1007/978-3-319-64419-6_13}. Two  related topics are the construction of minimal covers of languages, \cite{Domaratzki2001MinimalCO}, and  the  immunity  of languages, \cite{10.1007/978-3-662-21545-6_34,post1944,YAMAKAMI2013116}. Recall that a language $L_1$ is called $\mathcal{C}$-\emph{immune} if there is no infinite language $L_2\subseteq L_1$ such that $L_2\in \mathcal{C}$. 

Let $\mathbb{N}^+$ denote the set of all positive integers.
An infinite language $L$ is called \emph{constantly growing} if there is a constant $c$ such that for every word $u\in L$ there is a word $v\in L$ with $\vert u\vert<\vert v\vert\leq c+\vert u\vert$. In \cite{YAMAKAMI2013116}, it has been proved that every constantly growing language $L$ is $\REG$-dissectible.

We  introduce a ``natural''  generalization of constantly growing languages as follows.  
Let $\mathbb{R}^+$ denote the set of all positive real numbers. 
We define that a language $L$ is a \emph{geometrically growing} language if there is a constant $c\in\mathbb{R}^+$ such that for every $u\in L$ there exists $v\in L$ with $\vert u\vert<\vert v\vert\leq c\vert u\vert$. We say also  that $L$ is $c$-geometrically growing.
In the current article we show how to dissect geometrically growing language by a homomorphic image of intersection of two context-free languages.

Consider two alphabets $\Sigma$ and $\Theta$ such that $\vert \Sigma\vert=1$ and $\vert \Theta\vert=4$; that is $\Sigma$ and $\Theta$ denote alphabets with one letter and four letters, respectively. The main results of the current article are the following theorem and its corollary below.
\begin{theorem}
\label{dufife554d845}
There are context-free languages $\MM_1,\MM_2\subseteq \Theta^*$ and an erasing alphabetical homomorphism $\pi:\Theta^*\rightarrow \Sigma^*$ such that: If $L\subseteq\Sigma^*$ is a geometrically growing language then there is a regular language $R\subseteq \Theta^*$  such that $\pi\left(R\cap \MM_1\cap \MM_2\right)$ dissects the language $L$.
\end{theorem}

To emphasize the essential result of our article, we consider in Theorem \ref{dufife554d845} that $L$ is a language over an alphabet with one letter. Let $\Gamma$ denote a finite alphabet. The next corollary shows a generalization for a geometrically growing language over the alphabet $\Gamma$.
\begin{corollary}
There are context-free languages $\MM_1,\MM_2\subseteq \Theta^*$, an erasing alphabetical homomorphism $\pi:\Theta^*\rightarrow \Sigma^*$, and a nonerasing alphabetical homomorphism $\varphi : \Gamma^*\rightarrow \Sigma^*$ such that: If $L\subseteq \Gamma^*$ is a geometrically growing language then there is a regular language $R\subseteq \Theta^*$ such that $\varphi^{-1}\left(\pi\left(R\cap \MM_1\cap \MM_2\right)\right)$ dissects the language $L$.
\end{corollary}
\begin{proof}
Let $L_1 \subseteq \Gamma^*$ be a language and let $\varphi:\Gamma^*\rightarrow \Sigma^*$ be an alphabetical homomorphism defined as follows: $\varphi(a)=z$ for every $a\in\Gamma$, where $z$ is the only letter of the alphabet $\Sigma$. If $L_2\subseteq\Sigma^*$ and $L_2$ dissects $\varphi(L_1)$ then clearly the language $L_3=\varphi^{-1}(L_2)=\{w\in L_1\mid \varphi(w)\in L_2\}$ dissects $L_1$. This completes the proof.
\end{proof}

\begin{remark}
Since the intersection of a regular language and a context-free language is a context-free language we have $R\cap\MM_1\cap\MM_2$ is also  intersection of two context-free languages. This explains why we do not mention the regular language in the title of the article.
\end{remark}

We sketch the basic ideas of our proof. Note that a non-associative word on the letter $a$ is a ``well parenthesized'' word containing a given number of occurrences of $a$. It is known that the number of non-associative words containing $n+1$ occurrences of $a$ is equal to the $n$-th Catalan number \cite{stanley_fomin_1999}. For example for $n=3$ we have five distinct non-associative words: $(((aa)a)a)$, $((aa)(aa))$, $(a(a(aa)))$, $(a((aa)a))$, and $((a(aa))a)$. Every non-associative word contains the prefix $(^ka$ for some $k\in\mathbb{N}^+$, where $(^k$ denotes the $k$-th power of the opening bracket. We show that there are non-associative words such that $k$ equals ``approximately'' $\log_2{n}$. We construct two context-free languages whose intersection accepts such words and we call these words \emph{balanced extended non-associative words}. By counting the number of opening brackets of a balanced extended non-associative word with $n$ occurrences of $a$ we can compute the logarithm of the number of occurrences of $a$. 
If $L$ is a geometrically growing language then the language \[\widehat L=\{a^j\mid j=\lceil\log_2{\vert w\vert}\rceil\mbox{ and }w\in L\}\] is obviously constantly growing. Hence, by means of intersection of two context-free languages we transform the challenge of dissecting a geometrically growing language to the challenge of dissecting a constantly growing language. This approach allows us to prove our result.

\section{Preliminaries}
Let $\epsilon$ denote the empty word. Given a finite alphabet $\Alpha$, let $\Alpha^+$ denote the set of all finite nonempty words over the alphabet $\Alpha$ and let $\Alpha^*=\Alpha^+\cup\{\epsilon\}$.

Let $\Factor(w)$ denote the set of all factors of the word $w\in \Alpha^*$. We have $\epsilon,w\in\Factor(w)$. 

Let $\Prefix(w), \Suffix(w)\subseteq \Factor(w)$ denote the set of all prefixes and suffixes of $w\in \Alpha^*$, respectively. We have $\epsilon,w\in \Prefix(w)\cap\Suffix(w)$. Let $\occur(w,t)$ denote the number of occurrences of the factor $t\in \Alpha^+$ in the word $w\in \Alpha^*$; formally \[\occur(w,t)=\vert \{v\in\Suffix(w)\mid t\in \Prefix(v)\}\vert\mbox{.}\]

Given two finite alphabets $\Alpha_1,\Alpha_2$, a \emph{homomorphism} from $\Alpha_1^*$ to $\Alpha_2^*$ is a function $\tau:\Alpha_1^*\rightarrow \Alpha_2^*$ such $\tau(uv)=\tau(u)\tau(v)$, where $u,v\in \Alpha_1^*$. 
It follows that in order to define a homomorphism $\tau$, it suffices to define $\tau(a)$ for every $a\in \Alpha_1$; such definition ``naturally'' extends to every word $u\in \Alpha_1^*$. We say that $\tau$ is an \emph{alphabetical homomorphism} if $\tau(a)\in \Alpha_2$ for every $a\in \Alpha_1$. We say that $\tau$ is an \emph{erasing alphabetical homomorphism} if $\tau(a)\in \Alpha_2\cup\{\epsilon\}$ for every $a\in \Alpha_1$ and there is at least one $a\in \Alpha_1$ such that $\tau(a)=\epsilon$.

\section{Balanced non-associative words}

Let $\Theta=\{x,y,z,p\}$. We reserve the symbols $x,y,z,p$ for the  letters of the alphabet  $\Theta$. It means that wherever in our article we use the symbols $x,y,z,p$, we refer to the letters of $\Theta$.

Let $\ENW\subseteq \Theta^*$ be the language generated by the following context-free grammar, where $S$ is a start non-terminal symbol, $P$ is a non-terminal symbol, and $x,y,z,p\in\Theta$ are terminal symbols:
\begin{itemize}
\item $\Start\rightarrow xPPy$,
\item $P\rightarrow S \mid pzp\mid pzzp$.
\end{itemize}
We call the words from $\ENW$ \emph{extended non-associative words}. 
\begin{remark}
\label{udd887fjbns}
Let the letter $x$ represent an opening bracket and the letter $y$ a closing bracket.
It is easy to see that if $v_1,v_2\in\{pzp,pzzp\}\cup\ENW$ then $xv_1v_2y\in\ENW$.
Note that if $w\in \ENW$ then $w$ is ``well parenthesized'' with brackets $x$ and $y$.
Also note that if $w\in \ENW$, $xvy\in \Factor(w)$, $\occur(v,x)=0$, and $\occur(v,y)=0$, then $v\in\{pzppzzp, pzppzp, pzzppzzp, pzzppzp\}\mbox{.}$
\end{remark}

\begin{remark}
\label{dyeu788hf}
Recall from \cite{stanley_fomin_1999} that a ``standard'' non-associative word  on the letter $a$, mentioned in the introduction, can be represented as a full binary rooted tree, where every inner node represents a corresponding pair of brackets and every leaf represents the letter $a$. It is known that the number of inner nodes plus one is equal to the number of leaves in a full binary rooted tree. 

Obviously we can also represent the extended non-associative words from $\ENW$ as full binary rooted trees, where the factors $pzp$ and $pzzp$ represent the leaves. 
It follows that 
if $w\in \ENW$ then \[\occur(w,x)+1=\occur(w,pzp)+\occur(w,pzzp)\mbox{.}\]
If $w$ is a non-associative word on the symbol $a$ with brackets $x,y$ having $n+1$ occurrences of $a$ then we get $2^{n+1}$ extended non-associative words by replacing $a$ with $pzp$ or $pzzp$; for example if $K=\{xxa_1a_2ya_3y\mid a_1,a_2,a_3\in\{pzp,pzzp\} \}$ 
then $\vert K\vert=2^3=8$ and $K\subseteq \ENW$.
Since the number of non-associative words containing $n+1$ occurrences of $a$ is equal to the $n$-th Catalan number $C_n$ \cite{stanley_fomin_1999}, it is clear that 
\[\vert \{w\in \ENW\mid \occur(w,pzp)+\occur(w,pzzp)=n+1\}\vert=2^{n+1}C_n\mbox{,}\] where $n\in\mathbb{N}^+$.
\end{remark}

Let $\Balanced \subseteq \Theta^*$ be the language generated by the following context-free grammar, where $\Start$ is a start non-terminal symbol, $T,V,Z$ are  non-terminal symbols, and $x,y,z,p\in \Theta$ are terminal symbols:
\begin{itemize}
\item $\Start\rightarrow x\Start y \mid pZVZp$,
\item $V\rightarrow VZV \mid pTp$,
\item $T\rightarrow yTx\mid \epsilon$,
\item $Z \rightarrow  z\mid zz$.
\end{itemize}

We call the words from $\Balanced$ \emph{balanced words}. The reason for the name ``balanced'' comes from the following lemma.
\begin{lemma}
\label{fyth7xbsj}
If $u\in\Balanced$, $w\in \Factor(u)$, and $pwp\in \Factor(u)$ then $\occur(w,x)=\occur(w,y)$. 
\end{lemma}
\begin{proof}
The proof is by induction on $j=\occur(w,p)$. From the definition of the language $\Balanced$, it is clear that if $\occur(w,p)=0$ then $w=y^ix^i$ for some $i\in\{0\}\cup\mathbb{N}^+$. Hence we have the base case for $j=0$. Suppose $j>0$. Then it follows that $w= w_1 pw_2$ for some $w_1,w_2\in\Factor(w)$. Since $\occur(w_1,p)+1+\occur(w_2,p)=j$, we have $\occur(w_1,p),\occur(w_2,p)<j\mbox{.}$ Hence the lemma holds for both $pw_1p$ and $pw_2p$ and in consequence  lemma holds also for $pwp$. This completes the proof.
\end{proof}

Let $\BalNaw=\ENW\cap \Balanced\subseteq\Theta^*\mbox{.}$
We call the words from $\BalNaw$ \emph{balanced extended non-associative words}.
Let $\BalNaw(n)=\{w\in \BalNaw\mid \occur(w,z)=n\}$, where $n\in \mathbb{N}^+$. 

\begin{remark}
To understand the idea of balanced extended non-associative words, suppose $w\in\BalNaw$ and let $G$ be the full binary rooted tree that represents $w$ (as explained in Remark \ref{dyeu788hf}). Then in $G$, the length of the path from the root to a leaf does not depend on the leaf; it means the number of inner nodes lying on the path from a leaf to the root is a constant for $G$.
\end{remark}

\begin{example}
Let $v_1=xpzpxpzppzzpyy$ and $v_2=xxpzppzpyxpzzppzzpyy$. We have $v_1,v_2\in \ENW$, $v_1\not\in \BalNaw$,  and  $v_2\in \BalNaw$.
\end{example}

Given a word $w\in \Theta^*$, let $\height(w)=\max\{j\mid x^j\in \Factor(w)\}\mbox{.}$ We call $\height(w)$ the \emph{height} of $w$.
We show that if $w\in  \BalNaw$ and $h$ is the height of $w$ then $x^h$ is a prefix of $w$ and $y^h$ is a suffix of $w$.
\begin{lemma}
\label{oiu33888weu3}
If $w\in \BalNaw$ and $h=\height(w)$ then $x^h\in \Prefix(w)$ and $y^h\in \Suffix(w)$.
\end{lemma}
\begin{proof}
Since $\BalNaw\subseteq\ENW$, there is $\widehat h\in\mathbb{N}^+$ such that $x^{\widehat h}p\in \Prefix(w)$. To get a contradiction suppose that $\widehat h<h$. Because $\BalNaw\subseteq \Balanced$ it follows that $w=x^{\widehat h}pw_1py^hx^hpw_2$ for some $w_1\in \Factor(w)$, $w_2\in \Suffix(w)$, and \[\occur(x^{\widehat h}pw_1py^hx^h, x^h)=1\mbox{.}\]

\noindent
Lemma \ref{fyth7xbsj} implies that   $\occur(pw_1p,x)=\occur(pw_1p,y)$. Let $r=x^{\widehat h}pw_1py^h$. It follows that \[\occur(r,x)<\occur(r,y)\mbox{.}\]
This is a contradiction, since for every prefix $v\in \Prefix(w)$ of an extended non-associative word $w\in \ENW$ (a well parenthesized word) we have $\occur(v,x)\geq\occur(v,y)$. We conclude that $\widehat h=h$ and $x^h\in \Prefix(w)$. In an analogous way we can show that $y^h\in \Suffix(w)$. This completes the proof.
\end{proof}

For a word $w\in \BalNaw$, we show the relation between the height of $w$ and the number of occurrences of $z$ in $w$.
\begin{proposition}
\label{ryybvnmd73yrh}
If $w\in \BalNaw$ and $h=\height(w)$ then \[2^{h}\leq\occur(w,z)\leq 2^{h+1}\mbox{.}\]
\end{proposition}
\begin{proof}
From the definition of $\ENW$  it follows that $h\geq 1$.
We prove the proposition by induction. Obviously if $h=1$ then \[w\in\{xpzppzpy, xpzzppzpy, xpzppzzpy, xpzzppzzpy\}\mbox{.}\]
Thus the proposition holds for $h=1$. Suppose that the proposition holds for all $\widehat h<h$ and let $h\geq 2$. 
Since $\BalNaw\subseteq \ENW$,
 it follows that  
$h\geq 2$ implies that $w=xw_1w_2y$ for some $w_1, w_2\in \ENW\cup\{pzp,pzzp\}$ with $\{w_1,w_2\}\cap\ENW\not=\emptyset$. Without loss of generality suppose  that  $w_1\in\ENW$.

Let $h_1=\height(w_1)$.
Lemma \ref{oiu33888weu3} implies that $x^{h_1}\in \Prefix(w_1)$ and $y^{h_1}\in\Suffix(w_1)$. Lemma \ref{fyth7xbsj} implies that  $x^{h_1}\in\Prefix(w_2)$ and in consequence $w_2\in\BalNaw$, $y^{h_1}\in\Suffix(w_2)$, and $\height(w_2)=h_1$.

Because $x^{h_1}\in \Prefix(w_1)$ it follows that $x^{h_1+1}\in \Prefix(w)$. Thus  $h_1+1=h$. As we assumed that the proposition holds for all $\widehat h<h$, we can derive that
\[\occur(w,z)=\occur(w_1,z)+\occur(w_2,z)\leq 2^{h_1+1}+2^{h_1+1}=2^{h_1+2}=2^{h+1}\]
and
\[\occur(w,z)=\occur(w_1,z)+\occur(w_2,z)\geq 2^{h_1}+2^{h_1}=2^{h_1+1}=2^{h}\mbox{.}\]
This completes the proof.
\end{proof}
\begin{remark}
    Proposition \ref{ryybvnmd73yrh} could be also proven using tree graphs as follows: There are exactly $2^h$ leaves in a complete tree of height $h$, since there is a bijection (as
mentioned in Remark \ref{dyeu788hf}), there are $2^h$ occurrences of $pzp$ and $pzzp$, and hence the number of occurrences of $z$ in $w$ is between $2^h$ and $2^{h+1}$.
\end{remark}
Proposition \ref{ryybvnmd73yrh} has the following obvious corollary. 
\begin{corollary}
\label{rryt7fufyu3y}
If $n\in\mathbb{N}^+$, $w\in \BalNaw(n)$, and $h=\height(w)$ then
\[\log_2{n}-1\leq h\leq \log_2{n}\mbox{.}\]
\end{corollary}
\begin{remark}
Note that the number of occurrence of $z$ does not uniquely determine the height. For example if $w_1=xxpzppzpyxpzppzpyy$ and $w_2=xpzzppzzpy$ , then $w_1,w_2\in\BalNaw(4)$ and $1=\height(w_2)<\height(w_1)=2$. 
\end{remark}

Given $w,u,v\in \Theta^+$, let $\replace(w,v,u)$ denote the word built from $w$ by replacing the first occurrence of $v$ in $w$ by $u$. Formally, if $\occur(w,v)=0$ then $\replace(w,v,u)=w$. If $\occur(w,v)=j>0$ and $w=w_1vw_2$, where $\occur(vw_2,v)=j$ then $\replace(w,v,u)=w_1uw_2$.

We prove that the set of balanced extended non-associative words $\BalNaw(n)$ having $n$ occurrences of $z$ is nonempty for each  $n\geq 2$.
\begin{proposition}
\label{un7bbv3b}
If $n\in \mathbb{N}^+$ and $n\geq 2$ then $\BalNaw(n)\not=\emptyset$.
\end{proposition}
\begin{proof}
Let $j\in\mathbb{N}^+$ be such that $2^{j-1}<n\leq 2^j$. Obviously such $j$ exists and is uniquely determined. 
Let $w_1=xpzzppzzpy$. Let $w_{i+1}=xw_iw_iy$ for every $i\in\mathbb{N}^+$. Clearly $\occur(w_j,z)=2^{j+1}$ and $w_j\in\BalNaw(2^j)$.
Note that $\occur(w_j,pzzp)=2^{j}$. 

Let $w_{j,0}=w_j$ and $w_{j,i+1}=\replace(w_{j,i},pzzp,pzp)$, where $i\in \mathbb{N}^+\cup\{0\}$ and $i< 2^{j}$. Let $\alpha=2^j-n$. Then one can easily verify that $\occur(w_{j,\alpha},z)=n$ and $w_{j,\alpha}\in\BalNaw(n)$.

In principle, we construct a balanced extended non-associative word $w_j$ having $2^{j}$ occurrences of $pzzp$ and then we replace a certain number of occurrences of $pzzp$ with the factor $pzp$ to achieve the required number of occurrences of $z$. 
This completes the proof.
\end{proof}

\section{Dissection of infinite languages}

In \cite{YAMAKAMI2013116} it was shown that every constantly growing language can be dissected by some regular language. 
\begin{lemma}(see \cite[Lemma $3.3$]{YAMAKAMI2013116})
\label{dyf7u38udi}
Every infinite constantly growing language is $\REG$-dissectible.
\end{lemma}

In the next proposition we show under which condition we can dissect a language $L\subseteq\BalNaw$ by a regular language. Informally,  the proposition says that a geometrically growing subset of balanced extended non-associative words is $REG$-dissectible.
\begin{proposition}
\label{rujfkie778d}
If $\beta\in \mathbb{N}^+$, $\beta\geq \gmTwo$, and $L\subseteq \BalNaw$ is an infinite language such that for each $w_1\in L$ there is $w_2\in L$ with \[\occur(w_1,z)<\occur(w_2,z)\leq \beta\occur(w_1,z)\] then there is a regular language $R\subseteq \Theta^*$ such that $R$ dissects $L$.
\end{proposition}
\begin{proof}
Let $\alpha\in\mathbb{N}^+$ be such that $\gmTwo^{\alpha-1}< \beta\leq \gmTwo^{\alpha}$. Obviously such $\alpha$ exists and is uniquely determined.
Given $w_1\in L$, let $w_2\in L$ be such that \begin{equation}\label{ryd7wuy37y}n_1<n_2\leq \beta n_1\leq\gmTwo^{\alpha}n_1\mbox{,}\end{equation}where $n_1=\occur(w_1,z)$ and $n_2=\occur(w_2,z)$. From the conditions of the proposition such $w_2$ exists.

Without loss of generality suppose that $\log_2{n_1}\geq 2$. Note that there are only finitely many words $v\in\BalNaw$ with $\log_2{(\occur(v,z))}<2\mbox{.}$

Let $h_1=\height(w_1)$ and $h_2=\height(w_2)$. Corollary \ref{rryt7fufyu3y} implies that  \begin{equation}\label{ty7uyr7uq33}\log_2{n_1}-1\leq h_1 \mbox{ and }h_2\leq \log_2{n_2}\end{equation}

From (\ref{ryd7wuy37y}) and (\ref{ty7uyr7uq33}) it follows that
\begin{equation}
\label{rr6fhbsve}
h_2\leq \log_2{n_2}\leq  \log_2{(\gmTwo^{\alpha}n_1)}=\alpha+\log_2{n_1}\leq \alpha+1+h_1\mbox{.}
\end{equation}
Since we selected $w_1$ arbitrarily, it  follows from (\ref{rr6fhbsve}) that  \[H=\{x^h\mid h=\height(v)\mbox{ and } v\in L\}\] is a constantly growing language. 

Lemma \ref{dyf7u38udi} implies that $H\subseteq\{x\}^*$ is $REG$-dissectible. Let $\widehat R\subseteq \{x\}^*$ be a regular language that dissects $H$. Let $R=\{rpv\mid r\in\widehat R\mbox{ and }v\in\Theta^*\}\subseteq\Theta^{*}$. Obviously  $R$  is a regular language that dissects $L$; to see this, recall that if $w\in\BalNaw$, $i\in\mathbb{N}^+$, $a\in\Theta\setminus\{x\}$, and  $x^ia\in\Prefix(w)$ then $a=p$.

This completes the proof.
\end{proof}

We step to the proof of the main theorem of the current article.

\begin{proof}[of Theorem \ref{dufife554d845}]
Without loss of generality let $\Sigma=\{z\}$ be the alphabet with the letter $z\in\Theta$. Let $\pi: \Theta^*\rightarrow \Sigma^*$ be an erasing alphabetical homomorphism defined as follows:
\[
  \pi(a) =
  \begin{cases}
    z & \mbox{ If }a=z\mbox{.} \\
    \epsilon & \mbox{ If }a\in\{x,y,p\}\mbox{.}
  \end{cases}
\]

Thus $\pi$ erases all letters except for the letter $z$.
From the definition of $\ENW$, $\Balanced$, and $\BalNaw$, it follows that the language $\BalNaw$ is an intersection of two context-free languages $\ENW$ and $\Balanced$. Let $\MM_1=\ENW$ and let $\MM_2=\Balanced$; recall that $\MM_1$ and $\MM_2$ are used in the statement of Theorem \ref{dufife554d845}.

Let $\widehat L=\{w\in \BalNaw\mid \pi(w)\in L)\}\mbox{.}$ Note that $\widehat L$ contains $w\in \BalNaw$ if and only if there is a word $v\in L$ such that the number of occurrences of $z$ in $w$ is equal to the length of $v$; formally $\occur(w,z)=\vert v\vert$. Proposition \ref{un7bbv3b} implies that $\widehat L$ is an infinite language.  

Let $c\in\mathbb{R}^+$ be such that for every $u\in L$ there exists $v\in L$ with $\vert u\vert<\vert v\vert\leq c\vert u\vert$. Since $L$ is a geometrically growing language, we know that such $c$ exists. Let $\beta\in\mathbb{N}^+$ be such that $\beta\geq 2$ and $\beta\geq c$.
Hence $L$ is $\beta$-geometrically growing language. It follows  that if $w_1\in \widehat L$ then there is a word $w_2\in \widehat L$ with \[\occur(w_1,z)<\occur(w_2,z)\leq \beta\occur(w_1,z)\mbox{.}\] Then Proposition \ref{rujfkie778d} implies that there is a regular language $R\subseteq \Theta^*$ that dissects $\widehat L$.
This implies that the homomorphic image $\pi\left(R\cap\ENW\cap\Balanced\right)$ dissects the language $L$. 
This completes the proof.
\end{proof}

\section{Dissection and i-separation}
\label{skdjiek99f8f}
As mentioned in the introduction, for convenience we present a generalization of Lemma \ref{djueir8887d} (\cite[Lemma $5.1$]{YAMAKAMI2013116}), which demonstrates the connection between dissectibility and i-separation. The presented proof is just a copy of the proof in \cite{YAMAKAMI2013116} by changing $\REG$ to $\mathcal{C}$.
\begin{lemma}
\label{ujdkeijfi4rekd}
Let $\mathcal{A}$, $\mathcal{B}$, and $\mathcal{C}$ be any three language families and assume that $\mathcal{A}-\mathcal{B}$ is $\mathcal{C}$-dissectible. It then holds that, for any $A\in\mathcal{A}$ and any $B\in\mathcal{B}$, if $A$ i-covers $B$, then there exists a language in $\mathcal{E}$ that i-separates $i(B,A)$, where $\mathcal{E}$ expresses the set $\{B\cup(C\cap A)\mid A\in\mathcal{A}, B\in\mathcal{B},C\in\mathcal{C}\}$. In other words, $\mathcal{E}$ i-separates $i(\mathcal{B}, \mathcal{A})$.
\end{lemma}
\begin{proof}
Let $A\in\mathcal{A}$ and $B\in\mathcal{B}$ be two infinite languages. Let $D=A-B$ and assume that $D$ is infinite. Our assumption guarantees the existence of a language $C\in\mathcal{C}$ for which $C$ dissects $D$. We set $E=B\cup(A\cap C)$. Since $C$ dissects $D$, it follows that $\vert(A\cap C)-B\vert=\infty$ and $\vert(A\cap \overline C)-B\vert=\infty$. It follows that $B\subseteq E\subseteq A$ and $\vert A-E\vert=\vert E-B\vert=\infty$. Thus, $E$ i-separates $i(B,A)$. Since $C\in\mathcal{C}$, $E$ belongs to the language family $\mathcal{E}$. This completes the proof.
\end{proof}

\section{Open questions}

In the current article we applied a new idea of dissecting a language $L$ by a homomorphic image of a language $\widehat L$ from the family of languages $\CFL_2$. The idea can be generalized for every $\CFL_k$, where $k\in\mathbb{N}^+$. 
Let us introduce a notation for this technique.
Given an alphabet $\Alpha$ and a positive integer $k$, let 
\[\begin{split}\hiCFL_{k,\Alpha}=\{L\subseteq \Alpha^*\mid \mbox{ there are an alphabet }\widehat\Alpha\mbox{ and }\\ \mbox{ context-free languages }L_i\subseteq\widehat\Alpha^* \mbox{ with }i\in\{1,2,\dots,k\}\\ \mbox{ and a homomorphism }\phi:\widehat\Alpha^*\rightarrow\Alpha^* \\  \mbox{ such that }L=\phi(\bigcap_{i=1}^kL_i)\}\mbox{.}\end{split}\]

The prefix ``hi'' stands for ``homomorphic image''. Using the set $\hiCFL_{k,\Alpha}$ we can restate Theorem  \ref{dufife554d845} as follows:
\begin{corollary}
Every geometrically growing language over an alphabet $\Alpha$ with $\vert\Alpha\vert=1$ is $\hiCFL_{2,\Alpha}$-dissectible.
\end{corollary}

Moreover we introduced in the current article the notion of a geometrically growing language. We generalize this concept as follows.  
Let \[\Pi=\{\sigma:\mathbb{N}^+\rightarrow\mathbb{N}^+\mid \sigma(n)>n\mbox{ for all }n\in\mathbb{N}^+\}\mbox{.}\]
Given $\sigma\in\Pi$, we say that a language $L$ is \emph{$\sigma$-growing} if for every word $u\in L$ there is a word $v\in L$ such that $\vert u\vert<\vert v\vert\leq \sigma(\vert u\vert)$.
\begin{remark}
Let $\widetilde\sigma(n)=cn$ for some $c\in\mathbb{N}^+$ with $c>1$. Obviously $\widetilde\sigma\in\Pi$. A language $L$ is $\widetilde\sigma$-growing language if and only if $L$ is $c$-geometrically growing.
\end{remark}

Let $\mathcal{C}$ be a family of languages. Using the notion of $\sigma$-growing languages, we present the following open questions and problems:
\begin{itemize}
\item Find $\sigma\in\Pi$ such that there exists a $\sigma$-growing language $L$ that is not $\mathcal{C}$-dissectible or show that such $\sigma$ does not exist.
\item Find $\sigma\in\Pi$ such that: 
\begin{itemize}\item Every $\sigma$-growing language $L$ is $\mathcal{C}$-dissectible.
\item If $\widehat \sigma\in\Pi$ and $\widehat\sigma(n)>\sigma(n)$ for all $n\in\mathbb{N}^+$ then there is a $\widehat\sigma$-growing language $L$ that is not $\mathcal{C}$-dissectible.
\end{itemize}
\end{itemize}

Concerning the family of languages $\mathcal{C}$, we are particularly interested in $\REG$, $\CFL_k$, and $\hiCFL_{k,\Alpha}$ for all $k\in\mathbb{N}^+$.  However the questions may be of interest also for other families. 

We list some more open questions and problems in spite of the fact that some of them are already mentioned (directly or indirectly) above.
\begin{itemize}
\item Is the family of geometrically growing languages $\REG$-dissectible?
\item Does $\CFL$ i-separate $i(\CFL,\CFL)$ (mentioned in \cite{bucher1980} and \cite{YAMAKAMI2013116})?
\item Describe the ``minimal'' families of languages $\mathcal{C}$ such that $\mathcal{C}$ i-separates $i(\CFL,\CFL)$.
\item Let $\CFLGL\subseteq \CFL$ be the family of geometrically growing context-free languages. Describe the ``minimal'' families of languages $\mathcal{C}$ such that $\mathcal{C}$ i-separates $i(\CFLGL,\CFLGL)$. In particular, are all geometrically growing context-free
languages REG-dissectible?
\item Is there $\sigma\in\Pi$ such that if $i(L_1,L_2)\in i(\CFL,\CFL)$ then the language $L_2-L_1$ is $\sigma$-growing?
\item Describe languages that are $\CFL$-dissectible.
\item Find an example of a language $L$ such that $L$ is $\CFL$-dissectible and not $\REG$-dissectible.
\end{itemize}

For more open questions about dissectibility, we recommend the readers to review the section ``6. Future challenges'' in \cite{YAMAKAMI2013116}.

\acknowledgements
\label{sec:ack}
This work was supported by the Grant Agency of the Czech Technical University in Prague, grant No. SGS20/183/OHK4/3T/14.

\nocite{*}
\bibliographystyle{abbrvnat}
\bibliography{biblio}
\label{sec:biblio}

\end{document}